\documentclass[a4paper,USenglish,10p]{article}
\usepackage{etex}
\usepackage{amssymb,amsmath,amsthm}
\usepackage[vlined,ruled,linesnumbered]{algorithm2e}
\usepackage[utf8]{inputenc} 
\usepackage[T1]{fontenc}
\usepackage{pgf}
\usepackage{tikz}
\usepackage[sort,comma,square,numbers]{natbib}
\usepackage{xspace}
\usepackage{todonotes}
\usepackage[notcite,notref,final]{showkeys}
\usepackage{paralist}
\usepackage{textcomp} %
\usepackage{mathtools}

\usepackage[bookmarks=false, pdftex]{hyperref}

\title{Listing Conflicting Triples in Optimal Time}

\usepackage{authblk}
\usepackage[margin=25mm]{geometry}

\author{Mathias Weller}
\affil{CNRS, LIGM, Université Paris Est, Marne-la-Vallée, France}

\usetikzlibrary{arrows,shapes,positioning,calc,decorations.pathreplacing,decorations.markings,decorations.pathmorphing,patterns}
\pgfdeclarelayer{background2}
\pgfdeclarelayer{background}
\pgfdeclarelayer{foreground}
\pgfsetlayers{background2,background,main,foreground}

\tikzstyle{vertex}=[circle, draw, fill=white]
\tikzstyle{reti}=[vertex, fill=gray]
\tikzstyle{leaf}=[vertex, rectangle]
\tikzstyle{small}=[inner sep=2pt]
\tikzstyle{smallvertex}=[vertex, small]
\tikzstyle{smallleaf}=[leaf, small]
\tikzstyle{smallreti}=[reti, small]
\tikzstyle{edge}=[draw,-]
\tikzstyle{matching}=[edge,line width=3pt]
\tikzstyle{solution}=[gray!60, line width=5pt]
\tikzstyle{arc}=[draw,decoration={markings,mark=at position 1 with {\arrow[scale=.8]{latex}}}, postaction={decorate}]
\tikzstyle{revarc}=[draw, decoration={markings,mark=at position 0 with {\arrow[scale=.8,rotate=180]{latex}}}, postaction={decorate}]
\tikzstyle{bold}=[draw, line width=2pt]
\tikzstyle{boldarc}=[bold, arc]
\tikzstyle{optional}=[dashed]
\tikzstyle{path}=[decorate, decoration={snake, amplitude=.6mm}]

\newcommand{\drawchildren}[5]{
  \node[smallvertex, #4] (\detokenize{#1}0) at ($(#1)-(45:#2)$) {} edge[revarc] (#1);
  \node[smallvertex, #5] (\detokenize{#1}1) at ($(#1)-(135:#3)$) {} edge[revarc] (#1);
}
\newcommand{\drawleaves}[5]{
  \drawchildren{#1}{#2}{#3}{leaf, #4}{leaf, #5}
}

\colorlet{darkgreen}{green!50!black}
\colorlet{medgray}{gray!75}
\colorlet{lightgray}{gray!30}
\definecolor{linkcol}{rgb}{0,0,0.4} 
\definecolor{citecol}{rgb}{0.5,0,0} 

\ifx\newdefinition\undefined
  \let\newdefinition\newtheorem
\fi
\ifx\definition\undefined
  \newdefinition{definition}{Definition}
  \newtheorem{theorem}{Theorem}
  \newtheorem{lemma}{Lemma}
  
\fi

\newtheorem{observation}{Observation}

\newtheorem{obs}{Observation}

\newdefinition{construction}{Construction}

\ifx\qedhere\undefined
  \newcommand\qedhere{\hfill\qed}
\fi

\usepackage{float}
\floatstyle{plain}
\newfloat{multialg}{thp}{lop}
\floatname{multialg}{Algorithm}

\newcommand{\dist}[2][]{\ensuremath{\operatorname{dist}}\ifx\relax#1\relax\else\ensuremath{_{#1}}\fi\ensuremath{(#2)}}
\newcommand{\adist}[2][]{\ensuremath{\overline{\operatorname{dist}}}\ifx\relax#1\relax\else\ensuremath{_{#1}}\fi\ensuremath{(#2)}}
\renewcommand{\deg}[2][]{\ensuremath{\operatorname{deg}}\ifx\relax#1\relax\else\ensuremath{_{#1}}\fi\ensuremath{(#2)}}
\newcommand{\nh}[2][]{\ensuremath{\operatorname{N}}\ifx\relax#1\relax\else\ensuremath{_{#1}}\fi\ensuremath{(#2)}}
\newcommand{\inc}[2][]{\ensuremath{\operatorname{inc}}\ifx\relax#1\relax\else\ensuremath{_{#1}}\fi\ensuremath{(#2)}}
\newcommand{\LCA}[2][]{\ensuremath{\operatorname{LCA}}\ifx\relax#1\relax\else\ensuremath{_{#1}}\fi\ensuremath{(#2)}}

\newcommand{\com}[2]{\ensuremath{#1 \sqcap #2}}
\newcommand{\unc}[2]{\ensuremath{#1 \wr #2}}

\def\T{\mathcal{T}}
\newcommand\lind[2]{\ensuremath{#1|_{#2}}}

\newcommand{\leaves}{\ensuremath{\mathcal{L}}}

\let\emptyset\varnothing

\ifx\backref\undefined
\else
\renewcommand*{\backref}[1]{}
\renewcommand*{\backrefalt}[4]{%
\ifcase #1 %
(Not cited.)%
\or
(Cited on page~#2.)%
\else
(Cited on pages~#2.)%
\fi}

\fi

\makeatletter
\def\NAT@spacechar{~}
\makeatother

\newcommand{\raus}[1]{}

\newcommand{\probdef}[5]{
\hbox{\vbox{
\begin{quote}
  \label{#5}
  \ifthenelse{\equal{#3}{}}{}{{#3}\ifthenelse{\equal{#4}{}}{}{ ({#4})}}
  \begin{compactdesc}
    \item [Input:] {#1}
    \item [Question:] {#2}
  \end{compactdesc}
\end{quote}
}}
}

\newcommand{\taskprobdef}[5]{
\hbox{\vbox{
\begin{quote}
  \label{#5}
  \ifthenelse{\equal{#3}{}}{}{{#3}\ifthenelse{\equal{#4}{}}{}{ ({#4})}}
  \begin{compactdesc}
    \item [Input:] {#1}
    \item [Task:] {#2}
  \end{compactdesc}
\end{quote}
}}
}

\newcommand{\paraproblem}[6]{
\hbox{\vbox{
\begin{quote}
  \label{#6}
  \ifthenelse{\equal{#4}{}}{}{{#4}\ifthenelse{\equal{#5}{}}{}{ ({#5})}}
  \begin{compactdesc}
    \item [Input:] {#1}
    \item [Question:] {#2}
    \item [Parameter:] {#3}
  \end{compactdesc}
\end{quote}
}}
}

\def\nmid{\hspace{-1pt}{\not|}}

\hypersetup
{
bookmarksopen=false,
pdftoolbar=false, %
pdfmenubar=true, %
pdfhighlight=/O, %
colorlinks=true, %
pdfpagemode=UseNone, %
pdfpagelayout=SinglePage, %
pdffitwindow=true, %
linkcolor=linkcol, %
citecolor=citecol, %
urlcolor=linkcol %
}

\makeatletter
\renewcommand\bibsection%
{
  \section*{\refname
    \@mkboth{\MakeUppercase{\refname}}{\MakeUppercase{\refname}}}
}
\makeatother

\begin{document}

\maketitle

\begin{abstract}
  Different sources of information might tell different stories about the evolutionary history of a given set of species.
  This leads to (rooted) phylogenetic trees that ``disagree'' on triples of species, which we call ``conflict triples''.
  An important subtask of computing consensus trees which is interesting in its own regard is the enumeration of all
  conflicts exhibited by a pair of phylogenetic trees (on the same set of $n$ taxa).
  As it is possible that a significant part of the $\binom{n}{3}$ triples are in conflict,
  the trivial $\theta(n^3)$-time algorithm that checks for each triple whether it constitutes a conflict, was considered optimal.
  It turns out, however, that we can do way better in the case that there are only few conflicts.
  In particular, we show that we can enumerate all $d$ conflict triples between a pair of phylogenetic trees in $O(n+d)$ time.
  Since any deterministic algorithm has to spend $\Theta(n)$ time reading the input and $\Theta(d)$ time writing the output,
  no deterministic algorithm can solve this task faster than we do (up to constant factors).
\end{abstract}

\section{Introduction}\label{sec:intro}

In bioinformatics -- more precisely, phylogenetics -- evolutionary trees (``phylogenetic trees'')
are one of the fundamental types of data representation and, thus, among the most important objects being algorithmically
analyzed and manipulated.
A phylogenetic tree visualizes the evolutionary history of a set of taxa (e.g.\ a family of genes, a collection of species, etc.).
However, different sources of information might imply different evolutionary histories of the same taxa.
Such contradictions manifest themselves as ``conflict triples'' (sometimes also ``conflict triplets''), that is,
three taxa, say~$a$, $b$, and~$c$ such that
one phylogenetic tree~$P$ implies that a common ancestor of $a$ and $b$ split off the common lineage of $a$, $b$ and $c$ before splitting into $a$ and $b$
while another tree~$Q$ implies that a common ancestor of $b$ and $c$ split off the common lineage before splitting into $b$ and $c$.
More formally, $\LCA[P]{ab}\ne\LCA[P]{abc}$ and $\LCA[Q]{bc}\ne\LCA[Q]{ab}=\LCA[Q]{abc}$.
See \autoref{fig:conflict} for an example.

\begin{figure}[t]
  \centering
  \begin{tikzpicture}[scale=.3]
    \node at (0,3) {$P$};
    \foreach \i/\l in {0/A, 1/B, 2/C, 3/D, 4/E} \node[smallleaf, label=below:$\l$] (l\i) at (2*\i,0) {};
    \node[smallvertex] (v0) at (1,1) {} edge[arc] (l0) edge[arc] (l1);
    \node[smallvertex] (v1) at (5,1) {} edge[arc] (l2) edge[arc] (l3);
    \node[smallvertex] (v2) at (6,2) {} edge[arc] (v1) edge[arc] (l4);
    \node[smallvertex] (v3) at (4,4) {} edge[arc] (v0) edge[arc] (v2);
  \end{tikzpicture}
  \hspace{10mm}
  \begin{tikzpicture}[scale=.3]
    \node at (0,3) {$Q$};
    \foreach \i/\l in {0/A, 1/B, 2/C, 3/D, 4/E} \node[smallleaf, label=below:$\l$] (l\i) at (2*\i,0) {};
    \node[smallvertex] (v0) at (1,1) {} edge[arc] (l0) edge[arc] (l1);
    \node[smallvertex] (v1) at (7,1) {} edge[arc] (l3) edge[arc] (l4);
    \node[smallvertex] (v2) at (6,2) {} edge[arc] (v1) edge[arc] (l2);
    \node[smallvertex] (v3) at (4,4) {} edge[arc] (v0) edge[arc] (v2);
  \end{tikzpicture}
  \caption{Two phylogenetic trees $P$ and $Q$ with conflict $CDE$ (boxes = leaves, circles = inner vertices). In particular, $CD|_PE$ and $DE|_QC$.}
  \label{fig:conflict}
\end{figure}
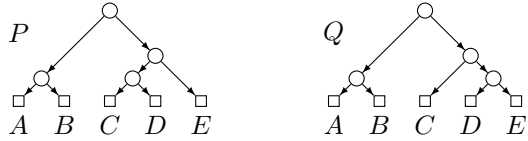

Conflict triples are essential ingredients to algorithms building so-called ``supertrees'', that is,
phylogenetic trees that merge evolutionary histories into one that is ``most consistent''~\cite{journal-JBS06,journal-BGJ10}.
Conflict triples can also be used to reconcile gene trees into a single phylogeny
by building a so-called ``triplet-based median supertree''~\cite{journal-RCD10}.
The problem of \emph{counting} conflict triples has been used to measure the distance between phylogenetic trees.
\citet{conf-BFM+13} show how to compute this number in $O(n\log n)$~time.
A recent study of the problem of finding a consensus tree given a set of disagreeing phylogenetic trees~\cite{CJL+17}
makes heavy use of the list of all conflict triples between any two of the input trees,
but does not detail how to enumerating them efficiently.
Here, we address this problem, showing how to enumerate all $d$ conflict triples of a pair $(P,Q)$ of phylogenetic trees on $n$ taxa in $O(n+d)$~time.
Since all algorithms solving this problem need to read the input (size $\Theta(n)$) and write the output (size $\Theta(d)$),
this is asymptotically ``best possible''.

While \emph{counting} the number of conflicts has received some attention in the past~\cite{conf-BFM+13},
not much work has been done on \emph{enumerating} them.
Such development might have been discouraged by the fact that a significant portion of the $\binom{n}{3}$ triples of taxa might be in conflict,
in which case the trivial algorithm that tests each triple of taxa for being a conflict would be optimal.
This work emerged from the question whether we can do better if only few triples are actually in conflict.
While preliminary works in this direction focussed on decision problems~\cite{conf-HW07, arxiv-MNN17, arxiv-FLP+15},
we consider an enumeration-type problem here.
Indeed, the concept of measuring the complexity in the size of the input \emph{and} the output
is fairly well known as \emph{output sensitivity} in the context of enumeration algorithms.
Running in $O(n+d)$ time where $n$ is the size of the input and $d$ is the size of the output,
our algorithm can be called \emph{totally linear}.

\section{Preliminaries}\label{sec:prelims}

A \emph{(phylogenetic) tree} is a rooted, binary%
\footnote{While we only consider binary phylogenetic trees in this work, I conjecture that it easily generalizes.}
outbranching whose leaves are bijectively labeled by a set $X$ (of taxa)
and we refer to its root by $r(T)$.
Since the labeling is bijective, we use leaves and labels interchangeably.
If some vertex $v$ of $T$ is a strict ancestor of a vertex $u$ in $T$, we write $u<_T v$ and we abbreviate $\mathop{\forall}_{v\in Z} v<_T u$ to $Z<_T u$.
We also abbreviate sets of leaves (or labels) by the concatenation of their names, that is, $abc$ refers to $\{a,b,c\}$.
The \emph{least common ancestor} of two leaves (or labels) $a$ and $b$ in $T$ is the minimum among all $u$ with $ab<_T u$ and we write $\LCA[T]{ab}=u$.
In this work a \emph{triple} $abc$ in $T$ is a set of three labels $abc\subseteq X$.
We say that $abc$ \emph{touches} $\LCA[T]{abc}$ and omit the mention of $T$ if it is clear from context.
We say a triple $abc$ is \emph{$ab$-biased} in $T$ if $\LCA[T]{ab}\ne\LCA[T]{abc}$ and we write $ab|_Tc$ to indicate this fact.
A triple $abc$ is called a \emph{conflict} of a pair $(P,Q)$ of trees if, for some $xy\subseteq abc$,
we have that $abc$ is $xy$-biased in exactly one of $P$ and $Q$ (see \autoref{fig:conflict}).
Recall that $abc$ and $cab$ refers to the same conflict, so when claiming that $abc$ is not listed twice,
this also means that no two permutations of $abc$ are listed.

For two vertices $u\in V(P)$ and $v\in V(Q)$, we define $\com{u}{v}:= \leaves(P_u)\cap \leaves(Q_v)$ and $\unc{u}{v}:=\leaves(P_u)\setminus\leaves(Q_v)$.
Note that $\com{}$ is symmetrical while $\unc{}$ is not.

\begin{obs}\label{obs:symmetry}
  Let $P$ and $Q$ be phylogenetic trees on the same leaf-set.
  Let $r_p$ and $r_q$ be the roots of $P$ and $Q$, respectively, and
  let $u_p$, $v_p$ and $u_q$, $v_q$ be their respective children.
  Then, $\unc{u_p}{u_q} = \com{u_p}{v_q} = \com{v_q}{u_p} = \unc{v_q}{v_p}$.
\end{obs}

In the following, we call a tree $T$ \emph{LCA-enabled} if the LCA of any two vertices in $T$ can be found in constant time.
Note that we can LCA-enable any tree in linear time~\cite{journal-HT84,conf-BF00}.

In the algorithm, we will want to compute the subtree $T'$ of a tree $T$ that is induced by a set $Z$ of leaves.
If $Z$ is ordered by an in-order or post-order traversal of $T$,
then this can be done in $O(|Z|)$ time~\cite[Section 8]{journal-CFR+00}.
The idea is that the inner vertices of $T'$ are exactly the LCAs of consecutive (wrt.\ the order) leaves in $Z$ and
the arcs between them can be computed by looking at the nearest, lower vertex on the left and right of each inner
vertex of $T'$ according to the order.

\begin{obs}[{\cite[Section 8]{journal-CFR+00}}]\label{lem:induced subtree}
  Let $T$ be an LCA-enabled tree
  and let $Z\subseteq\leaves(T)$ be in post-order.
  Then, $\lind{T}{Z}$ can be computed in $O(|Z|)$ time.
\end{obs}

Furthermore, for leaf-labelled trees $P$ and $Q$ and vertices $u$ and $v$ of $P$ and $Q$, respectively,
we will want to detect whether $\leaves(P_u)=\leaves(Q_v)$ in constant time.
To this end, we construct a mapping $m$ that maps each vertex $x$ of $P$ to the unique vertex $y$ of $Q$ that is lowest among
all vertices of $Q$ satisfying $\leaves(P_x)\subseteq\leaves(Q_y)$.
Note that, $m(x) = \LCA[Q]{m(x'), m(x'')}$ where $x'$ and $x''$ are the children of $x$ in $P$ and, thus,
$m$ can be computed in $O(|P|+|Q|)$ time if $Q$ is LCA-enabled.
Finally, we only need to know the number of leaves reachable from each vertex of $P$ and $Q$,
which can easily be computed in $O(|P|+|Q|)$ time.

\begin{obs}\label{obs:leaf-set equiv}
  Let $P$ and $Q$ be phylogenetic trees on the same leaf-set and
  let $Q$ be LCA-enabled.
  Then, there is a linear-time preprocessing that allows answering if $\leaves(P_u)=\leaves(Q_v)$ in constant time for each $u$ and $v$.
\end{obs}

\section{The Algorithm}\label{sec:algo}

Given two phylogenetic trees $P$ and $Q$ on the label-set $X$, our algorithm will first list all conflict triples $abc$ that touch $r(P)$ or $r(Q)$ and
then recurse into specific induced subtrees of $P$ and $Q$ such that,
the conflicts in these subtrees are exactly the conflicts between $P$ and $Q$ that do not touch $r(P)$ and $r(Q)$.
The observation that being a conflict triple is invariant under deletion of unrelated leaves implies the correctness of this approach.

\begin{multialg}[t]
  \begin{procedure}[H]
    \caption{ListCommonRootConflicts()}
    \KwIn{Trees $P$ \& $Q$ on $X$, a child $x_p$ of $r(P)$, a child $x_q$ of $r(Q)$}
    \KwOut{Conflict triples $abc$ with $ab\leq x_p$ touching $r(P)$ and $r(Q)$}
    \lForEach{$a\in\com{x_p}{x_q}$ and $b\in\unc{x_p}{x_q}$ and $c\in X\setminus\leaves(x_p)$}{%
      \textbf{list} $abc$%
    }
  \end{procedure}
  \begin{procedure}[H]
    \caption{ListUncommonRootConflicts()}
    \KwIn{Trees $P$ \& $Q$ on $X$, a child $x_p$ of $r(P)$, a child $x_q$ of $r(Q)$}
    \KwOut{Conflict triples $abc\leq x_p$ touching $r(Q)$ (but not $r(P)$)}
    \lForEach{$a,b\in\com{x_p}{x_q}$ and $c\in\unc{x_p}{x_q}$ with $ab\nmid_P c$}{%
      \textbf{list} $abc$%
    }
    \lForEach{$a,b\in\unc{x_p}{x_q}$ and $c\in\com{x_p}{x_q}$ with $ab\nmid_P c$}{%
      \textbf{list} $abc$%
    }
  \end{procedure}
  \begin{procedure}[H]
    \caption{ListAllConflicts()}
    \KwIn{Trees $P$ \& $Q$}
    \KwOut{Conflict triples of $(P,Q)$}
    \If{$|\leaves(P)|>1$}{
      $(u_p,u_q),(v_p,v_q)\gets$ arbitrary pairing of children of $r(P)$ \& $r(Q)$\;\nllabel{ln:pairs}
      \ForEach{$(x_p,x_q)\in \{(u_p,u_q),(v_p,v_q)\}$}{
        compute $\com{x_p}{x_q}$, $\unc{x_p}{x_q}$ and $\unc{x_q}{x_p}$\;
        $\ListCommonRootConflicts(P,Q,x_p,x_q)$\;
        $\ListUncommonRootConflicts(P,Q,x_p,x_q)$\;
        $\ListUncommonRootConflicts(Q,P,x_q,x_p)$\;
        $\ListAllConflicts(\lind{P}{\com{x_p}{x_q}}, \lind{Q}{\com{x_p}{x_q}})$\;
      }
      $\ListAllConflicts(\lind{P}{\unc{u_p}{u_q}}, \lind{Q}{\unc{v_q}{v_p}})$\;
      $\ListAllConflicts(\lind{P}{\unc{v_p}{v_q}}, \lind{Q}{\unc{u_q}{u_p}})$\;
    }
  \end{procedure}
  \caption[First shot at triplet enumeration]{First shot at triplet enumeration. Note that, although theoretically unnecessary, we provide $x_q$ to the calls to \ListCommonRootConflicts and \ListUncommonRootConflicts, since this lets us use the pre-computed  sets $\com{x_p}{x_q}$ and $\unc{x_p}{x_q}$ and $\unc{x_q}{x_p}$.\label{alg:first shot}}
\end{multialg}

\begin{observation}\label{lem:recurse}
  Let $Y\subseteq X$, and
  let $abc\subseteq Y$.
  Then,
  $abc$ is a conflict triple of $(P,Q)$
  if and only if
  $abc$ is a conflict triple of $(\lind{P}{Y},\lind{Q}{Y})$.
\end{observation}

\begin{observation}\label{obs:all conflicts}
  Let $abc$ be a conflict triple of $(P,Q)$ that touches neither $r(P)$ nor $r(Q)$.
  Let $u_p$ and $v_p$ be the children of $r(P)$ and let $u_q$ and $v_q$ be the children of $r(Q)$.
  Then, $abc$ is completely contained in $\com{u_p}{u_q}$, $\com{u_p}{v_q}$, $\com{v_p}{u_q}$, or $\com{v_p}{v_q}$.
\end{observation}

\noindent
Note that the four sets mentioned in \autoref{obs:all conflicts} are disjoint, and so,
no conflict can be contained in any two of them.
Then, our algorithm can be described as the following recursion (see \autoref{alg:first shot} for a detailed description):
\begin{description}
  \item[Base Case:]
    If $r(P)$ and $r(Q)$ are leaves, then return without listing anything.
  \item[Recursion:]
    First, choose an arbitrary pairing $\{(u_p,u_q), (v_p,v_q)\}$ of the children of $r(P)$ and $r(Q)$.
    Second, list all conflict triples $abc$ touching $r(P)$ or $r(Q)$.
    Third, recursively list all conflict triples of
    \begin{compactenum}
      \item $(\lind{P}{\com{u_p}{u_q}},\lind{Q}{\com{u_p}{u_q}})$,
      \item $(\lind{P}{\com{v_p}{v_q}},\lind{Q}{\com{v_p}{v_q}})$,
      \item $(\lind{P}{\com{u_p}{v_q}},\lind{Q}{\com{u_p}{v_q}})$ and
      \item $(\lind{P}{\com{v_p}{u_q}},\lind{Q}{\com{v_p}{u_q}})$.
    \end{compactenum}
\end{description}

\begin{procedure}[t]
  \caption{ListSubtreeConflicts()}
  \KwIn{Tree $T$, leaf subset $Z\subseteq\leaves(T)$ in post-order}
  \KwOut{Triples $abc$ with $a,b\in Z$, and $c\in \leaves(T)\setminus Z$, and $ab\nmid_T c$}

  \If{$Z\ne\emptyset$}{
    \ForEach{$c\in\leaves(T)\setminus Z$}{
      $T'\gets \lind{T}{Z\cup\{c\}}$\;\nllabel{ln:Tprime}
      $y\gets$ parent of $c$\;\nllabel{ln:start y}
      \While{$y\ne r(T')$\nllabel{ln:while}}{
        $y'\gets$ sibling of $y$ in $T'$\;
        \lForEach{$a\in \leaves(T'_y)\setminus\{c\}$ and $b\in\leaves(T'_{y'})$}{\textbf{list} $abc$}
        $y\gets$ parent of $y$ in $T'$\;\nllabel{ln:next y}
      }
    }
  }
\end{procedure}

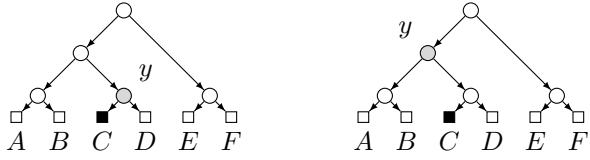
\begin{figure}[t]
  \centering
  \begin{tikzpicture}[scale=.4]
    \node[smallvertex] (v) {};
    \drawchildren{v}{2}{4}{}{}
    \drawchildren{v0}{2}{2}{}{fill=lightgray, label=above right:$y$}
    \drawleaves{v00}{1}{1}{label=below:$A$}{label=below:$B$}
    \drawleaves{v01}{1}{1}{fill=black,label=below:$C$}{label=below:$D$}
    \drawleaves{v1}{1}{1}{label=below:$E$}{label=below:$F$}
  \end{tikzpicture}
  \hspace{10mm}
  \begin{tikzpicture}[scale=.4]
    \node[smallvertex] (v) {};
    \drawchildren{v}{2}{4}{fill=lightgray, label=above left:$y$}{}
    \drawchildren{v0}{2}{2}{}{}
    \drawleaves{v00}{1}{1}{label=below:$A$}{label=below:$B$}
    \drawleaves{v01}{1}{1}{fill=black,label=below:$C$}{label=below:$D$}
    \drawleaves{v1}{1}{1}{label=below:$E$}{label=below:$F$}
  \end{tikzpicture}
  \caption[Illustration of T']{An example illustrating the tree $T'$ in two steps of \ListSubtreeConflicts (gray = vertex~$y$, black = leaf~$c$ with label~$C$).
    Left: first step ($y$ is the parent of $c$), listing $DAC$ and $DBC$.
    Right: second step, listing all $abC$, with $a\in\{A,B,D\}$ and $b\in\{E,F\}$.}
  \label{fig:LSC}
\end{figure}

We defer showing correctness in favor of introducing some modifications that allow achieving our running-time goal.
In order to see why this is necessary, let us analyze \ListAllConflicts.
This requires a closer look at how many triples are listed in each recursive step.
\ListCommonRootConflicts unconditionally lists
  $|\mbox{\com{x_p}{x_q}}|\cdot|\mbox{\unc{x_p}{x_q}}|\cdot |X\setminus\leaves(x_p)|$
conflicts for each pair $(x_p,x_q)$ of the chosen pairing.
However, \ListUncommonRootConflicts has to perform
numerous
checks of the type ``$ab|c$?''.
Since it is possible that none of these triples is a conflict, we cannot bound these operations in the number of listed conflicts.
Instead, we use \ListSubtreeConflicts to list all the triples $abc$ with $a,b\in\com{x_p}{x_q}$ and $c\in\unc{x_p}{x_q}$ (or vice versa),
and $ab\nmid_P c$ in constant time per listed triple (see \autoref{fig:LSC} for an illustration).
The idea is
\begin{inparaenum}[(i)]
  \item to focus on the subtree $P'$ of $P$ that is rooted at $\LCA[P]{\com{x_p}{x_q}}$,
  \item to pick any leaf $c\in\unc{x_p}{x_q}$ and,
  \item for each $y$ on the unique path from $c$ to $r(P')$, listing all triples $abc$ for which $a$ and $c$ are ``below $y$'' and $b$ is not,
    thereby ensuring $\LCA[T]{ac}\ne\LCA[T]{abc}$.
\end{inparaenum}
We will thus replace the first for-loop of \ListUncommonRootConflicts by a call to
$\ListSubtreeConflicts(P, \com{x_p}{x_q})$ and the second for-loop with a call to $\ListSubtreeConflicts(P, \unc{x_p}{x_q})$.

\begin{lemma}\label{lem:LSC correct}
  \ListSubtreeConflicts is correct, that is, it outputs a triple $abc$ if and only if $a,b\in Z$, $c\notin Z$, and $ab\nmid_T c$.
  Further, the procedure takes $O(d)$ time (where $d$ is the total number of listed triples) and no triple is listed twice.
\end{lemma}
\begin{proof}
  We first show the first equivalence.

  ``$\Rightarrow$'':
  Let $abc$ be a listed triple.
  Then, there is some $y$ with $c< y<r(T')$ 
  with sibling $y'$ such that $a\in \leaves(T'_y)\setminus\{c\}$ and $b\in\leaves(T'_{y'})$ (by symmetry among $ab$).
  But then, $a,b\in Z$, and $c\notin Z$ and $ac<_{T'} y$ and $b\leq_{T'} y'$,
  implying $ac|_{T'} b$ and, thus, $ac|_T b$.

  ``$\Leftarrow$'':
  Let $abc$ be a triple with $a,b\in Z$, $c\notin Z$ and $ab\nmid_T c$.
  Then, $|Z\ne\emptyset$, and $c\in\leaves(T)\setminus Z$.
  Since $ab\nmid_T c$, we have $\LCA[T]{ab}=\LCA[T]{abc}$ and,
  by symmetry among $ab$, we suppose $\LCA[T]{ac}<\LCA[T]{abc}$.
  Let $y$ and $y'$ be the children of $\LCA[T']{abc}$ with $a,c<_{T'} y$
  and note that $y$ will be reached by the while-loop.
  Clearly, $a\in\leaves(T'_y)$, and $b\in\leaves(T'_{y'})$, and, thus, $abc$ is listed.

  \smallskip
  Second, suppose that any triple $abc$ is listed twice.
  As $y$ and $y'$ are siblings in each iteration of the while-loop, $abc$ is listed for two different values of $y$.
  However, there is a single vertex (namely $\LCA{ab}$) for which neither $ab\subseteq\leaves(T'_y)$ nor $ab\subseteq\leaves(T'_{y'})$.
  Thus, there is a single iteration for which $abc$ can be output.

  \smallskip
  Finally, we show the claimed running time.
  We start by showing that, each time the while-loop is run, it outputs at least $|Z|-1$ triples.
  To this end, consider $y'$ and its sibling $y$ in any last iteration of the while-loop (that is, the parent of $y$ and $y'$ is $r(T')$).
  Then, the number of triples that are listed is 
  $|\leaves(T'_y)-1|\cdot|\leaves(T'_{y'})|\geq|\leaves(T'_y)|-1+|\leaves(T'_{y'}|-1=|\leaves(T')|-1=|Z|-1$.
  Since, by \autoref{lem:induced subtree}, $T'$ can be computed in $O(|Z|)$ time (line~\ref{ln:Tprime}),
  we conclude that \ListSubtreeConflicts runs in $O(d)$ time.
\end{proof}

With Lemma~\ref{lem:LSC correct}, we can finally list all $d_r$ conflict triples $abc$ with
$\LCA[P]{abc}=r(P)$ or $\LCA[Q]{abc}=r(Q)$ in $O(d_r)$ time.
Thus, \ListAllConflicts completes the following tasks in the mentioned times.
\begin{compactenum}[({Task}~a)]
  \item\label{it:list} list all conflict triples touching $r(P)$ or $r(Q)$: $O(d_r)$ time;
  \item\label{it:common} compute common and uncommon leaves: $O(|X|)$ time;
  \item\label{it:induce} compute the subtrees induced by these leaf-sets: $O(|X|)$ time;
  \item\label{it:prep} preprocess these subtrees for the recursive calls: $O(|X|)$ time;
  \item\label{it:recurse} make recursive calls
\end{compactenum}
The algorithm in its current form has a worst-case running time of $O(|X|^2)$.
In the following, we show how to avoid the costly computations of \eqref{it:common}, \eqref{it:induce}, and \eqref{it:prep} if they are unnecessary and bound their running-time in $O(d_r)$ if they cannot be avoided.
To this end, note that, when called with $u_p$ and $u_q$, \ListCommonRootConflicts outputs
\[
  |\com{u_p}{u_q}|\cdot|\unc{u_p}{u_q}|\cdot(|\com{v_p}{v_q}|+|\unc{v_p}{v_q}|) \leq d_r
\]
unique conflicts. Thus, if $\com{u_p}{u_q}\ne\emptyset$ and $\unc{u_p}{u_q}\ne\emptyset$, then
\begin{align*}
  |X| & = (|\com{u_p}{u_q}| + |\unc{u_p}{u_q}|) + (|\com{v_p}{v_q}|+|\unc{v_p}{v_q}|)\\
      & \leq |\com{u_p}{u_q}|\cdot|\unc{u_p}{u_q}|\cdot (|\com{v_p}{v_q}|+|\unc{v_p}{v_q}|) + 2 \leq d_r + 2
\end{align*}
and we can thus bound the time spent for \eqref{it:common}, \eqref{it:induce}, and \eqref{it:prep} in $O(d_r)$.
By symmetry, the same holds if $\com{v_p}{v_q}\ne\emptyset$ and $\unc{v_p}{v_q}\ne\emptyset$.
It remains to explore the cases that one of $\com{u_p}{u_q}$ and $\unc{u_p}{u_q}$ and one of $\com{v_p}{v_q}$ and $\unc{v_p}{v_q}$ is empty.
\begin{compactdesc}
  \item[First, $\unc{u_p}{u_q}=\com{v_p}{v_q}=\emptyset$.]
    Then all leaves of $P_{u_p}$ are leaves of $Q_{u_q}$ and all leaves of $P_{v_p}$ are not leaves of $Q_{v_q}$.
    Thus, $Q_{v_q}$ does not have any leaves, contradicting the fact that $P$ and $Q$ are binary trees.
    Symmetrically, $\com{u_p}{u_q}=\unc{v_p}{v_q}=\emptyset$ cannot happen.
  \item[Second, $\unc{u_p}{u_q}=\unc{v_p}{v_q}=\emptyset$.]
    Then, $\leaves(u_p)=\leaves(u_q)$ and $\leaves(v_p)=\leaves(v_q)$.
    This situation can be detected in constant time, given a linear-time preprocessing of $P$ and $Q$
    that links a node $x_p$ of $P$ to a node $x_q$ of $Q$ if and only if $P_{x_p}$ and $Q_{x_q}$ have the same leaf-set (see \autoref{obs:leaf-set equiv}).
    In this case, there are no root-conflicts and none of the costly steps \eqref{it:common}--\eqref{it:prep} are necessary.
  \item[Third, $\com{u_p}{u_q}=\com{v_p}{v_q}=\emptyset$.]
    Then, changing the root-child pairing to $(u_p,v_q)$ and $(v_p,u_q)$ gives the previous case.
    The same preprocessing allows us to detect and deal with this case.
\end{compactdesc}

\noindent
The final version of the algorithm is presented as \autoref{alg:final} and we can prove its running time and correctness.

\def\sameleaves{\ensuremath{\equiv_{\leaves}}}
\begin{multialg}[t]
  \small
  \caption[Refined algorithm to enumerate all conflict triples]{Refined algorithm to enumerate all conflict triples.
    Note that we do not have to update leaf-set equivalence relations for the recursions in lines~\ref{ln:rec u} and~\ref{ln:rec v}
    since the relation computed in the parent remains valid.\label{alg:final}}
  \begin{procedure}[H]
    \caption{ListAllConflicts'()}
    \KwIn{Trees $P$ \& $Q$, preprocessed to answer leaf-set equivalence in $O(1)$}
    \KwOut{Conflict triples of $(P,Q)$}
    $(u_p,u_q),(v_p,v_q)\gets$ arbitrary pairing of children of $r(P)$ \& $r(Q)$\;\nllabel{ln:choose childs}
    \lIf{$\leaves(u_p) = \leaves(v_q)$}{swap $u_q$ and $v_q$}\nllabel{ln:swap}
    \If{$\leaves(u_p) = \leaves(u_q)$\nllabel{ln:end const part}}{    
      $\ListAllConflicts'(P_{u_p}, Q_{u_q})$\;\nllabel{ln:rec u}
      $\ListAllConflicts'(P_{v_p}, Q_{v_q})$\;\nllabel{ln:rec v}
    }
    \Else{
      \ForEach{$(x_p,x_q)\in \{(u_p,u_q),(v_p,v_q)\}$}{
        compute \& post-order the sets $\com{x_p}{x_q}$, $\unc{x_p}{x_q}$ and $\unc{x_q}{x_p}$\;\nllabel{ln:comp sets}
        compute $\lind{P}{\com{x_p}{x_q}}$, $\lind{P}{\unc{u_p}{u_q}}$, $\lind{Q}{\com{x_q}{x_p}}$, and $\lind{Q}{\unc{x_q}{x_p}}$\;
        compute the leaf-set equivalence relation for corresponding tree-pairs\;\nllabel{ln:leaf equiv}
        $\ListCommonRootConflicts(P,Q,x_p,x_q)$\;\nllabel{ln:common}
        $\ListSubtreeConflicts(P,\com{x_p}{x_q})$\;\nllabel{ln:subtreeP1}
        $\ListSubtreeConflicts(P,\unc{x_p}{x_q})$\;\nllabel{ln:subtreeP2}
        $\ListSubtreeConflicts(Q,\com{x_q}{x_p})$\;\nllabel{ln:subtreeQ1}
        $\ListSubtreeConflicts(Q,\unc{x_q}{x_p})$\;\nllabel{ln:subtreeQ2}
        $\ListAllConflicts'(\lind{P}{\com{x_p}{x_q}}, \lind{Q}{\com{x_q}{x_p}})$\;\nllabel{ln:rec com}
      }
      $\ListAllConflicts'(\lind{P}{\unc{u_p}{u_q}}, \lind{Q}{\unc{v_q}{v_p}})$\;\nllabel{ln:rec unc}
      $\ListAllConflicts'(\lind{P}{\unc{v_p}{v_q}}, \lind{Q}{\unc{u_q}{u_p}})$\;
    }
  \end{procedure}
\end{multialg}

\begin{lemma}\label{lem:final alg correct}
  \autoref{alg:final} outputs a triple if and only if it is a conflict.
  Moreover, no conflict is listed twice and \autoref{alg:final} runs in $O(|X|+d)$ time,
  where $X$ is the label set of the input trees and $d$ is the total number of conflicts listed.
\end{lemma}
\begin{proof}
  Let line~\ref{ln:pairs} of \ListAllConflicts produce the pairs $(u_p,u_q)$ and $(v_p,v_q)$.

  ``$\Rightarrow$'':
  Let $abc$ be a triple that is listed by \autoref{alg:final}.
  If \ListCommonRootConflicts lists $abc$ then, without loss of generality,
  $a\in\com{u_p}{u_q}$, and $b\in\unc{u_p}{u_q}$, and $c\in X\setminus\leaves(u_p)$.
  Thus, $a\leq u_p,u_q$, and $b\leq u_p,v_q$, and $c\leq v_p$.
  Now, if $c\leq v_q$, then $ab|_P c$ and $a|_Q bc$, otherwise, $ab|_P c$ and $ac|_Q b$.
  In both cases, $abc$ is a conflict.
  Otherwise, $abc$ is listed by \ListSubtreeConflicts and, without loss of generality,
  let the first argument be $P$ (lines~\ref{ln:subtreeP1} and~\ref{ln:subtreeP2}).
  Then, by construction of \ListSubtreeConflicts, there is some
  $Z\in\{\com{x_p}{x_q},\unc{x_p}{x_q}\}$ and
  some $y$ such that
  $a,c<_P y$, and $a,b\in Z$, and $c\notin Z$, and $y<\LCA[P]{ab}$.
  Thus $ac|_P b$.
  Now, if $Z=\com{x_p}{x_q}$ then, as $c<y<x_p$ and $c\notin Z$, we have $c\nleq x_q$, but $a,b<x_q$, implying $ab|_Q c$.
  If $Z=\unc{x_p}{x_q}$ then, as $c<y<x_p$ and $c\notin Z$, we have $c\leq x_q$, but $a,b\nleq x_q$, implying $ab|_Q c$.
  In both cases, $abc$ is a conflict.

  ``$\Leftarrow$'':
  Let $abc$ be a conflict between $P$ and $Q$ and, 
  by symmetry among $abc$, let $ab|_P c$ and $ac|_Q b$.
  Further, by symmetry among $u_p$ and $v_p$, let $ab< u_p$.
  First, suppose that $\LCA[P]{abc}=r(P)$, that is, $c\leq v_p$.
  If $abc<u_q$ (or $abc<v_q$), then there is $Z:=\com{u_p}{u_q}$ (or $Z:=\unc{u_p}{u_q}$) with $a,b\in Z$ and $c\notin Z$ and $ab\nmid_Q c$ and,
  by Lemma~Lemma~\ref{lem:LSC correct}, $abc$ is listed by \ListSubtreeConflicts in line~\ref{ln:subtreeQ1} (or line~\ref{ln:subtreeQ2}).
  Otherwise, $\LCA[Q]{abc}=r(Q)$, that is, $ac<u_q$ and $b\leq v_q$ or vice versa (since $ac|_Q b$).
  But then, $ac<u_q$ (or $ac<v_p$) and $b\leq v_q$ (or $b\leq u_p$),
  implying $a\in\com{u_p}{u_q}$, and $b\in\unc{u_p}{u_q}$ (or $b\in\com{u_p}{u_q}$, and $a\in\unc{u_p}{u_q}$), and $c\not< u_p$ and,
  thus, $abc$ is listed by \ListCommonRootConflicts in line~\ref{ln:common}.
  Second, suppose that $\LCA[P]{abc}<r(P)$, that is, $c\leq u_p$.
  If $\LCA[Q]{abc}=r(Q)$, then $ac<u_q$ and $b<v_q$ or vice versa.
  But then, there is $Z:=\com{u_p}{u_q}$ (or $Z:=\unc{u_p}{u_q}$) with $a,c\in Z$, and $b\notin Z$ and $ac\nmid_P b$ and,
  by Lemma~\ref{lem:LSC correct}, $acb$ is listed by \ListSubtreeConflicts in line~\ref{ln:subtreeP1} (or line~\ref{ln:subtreeP2}).
  Otherwise, $\LCA[Q]{abc}<r(Q)$.
  If $abc<u_q$ then, by induction on the recursion depth,
  $abc$ is listed by the recursive call on line~\ref{ln:rec com} (or line~\ref{ln:rec u} if $\leaves(u_p)=\leaves(u_q)$).
  Otherwise, $abc<v_q$ and, by induction on the recursion depth,
  $abc$ is listed by the recursive call on line~\ref{ln:rec unc} (or line~\ref{ln:rec u} if $\leaves(u_p)=\leaves(v_q)$, as $u_q$ and $v_q$ would have been swapped in line~\ref{ln:swap} in this case).

  \smallskip
  To show that no conflict $abc$ is output twice, assume the contrary.
  Again, symmetry lets us suppose $ab|_P c$, and $ac|_Q b$, and $ab<u_p$.
  Note that the two occurrences of $abc$ cannot be output by
  \begin{compactitem}
    \item different recursive calls, since all tree-pairs in recursive calls have pairwise disjoint sets of leaf-labels,
    \item the same call to \ListCommonRootConflicts since $\com{x_p}{x_q}$, and $\unc{x_p}{x_q}$ and $X\setminus\leaves(x_p)$ are pairwise disjoint, or
    \item the same call to \ListSubtreeConflicts by Lemma~\ref{lem:LSC correct}.
  \end{compactitem}
  Thus, $abc$ is listed by different calls in the same node of the recursion tree.
  If $\LCA[P]{abc}=r(P)$ and $\LCA[Q]{abc}=r(Q)$, then $abc$ is listed by both calls to \ListCommonRootConflicts,
  implying that $abc$~intersects $\com{u_p}{u_q}$ and $\unc{u_p}{u_q}$ as well as $\com{v_p}{v_q}$ and $\unc{v_p}{v_q}$.
  However, as these sets are disjoint, this cannot happen.
  If $\LCA[P]{abc}=r(P)$ and $\LCA[Q]{abc}\ne r(Q)$, then $abc<_Q u_q$ or $abc<_Q v_q$ and $c\leq_P v_p$.
  If $abc<_Q u_q$, then $ab\subseteq\com{u_p}{u_q}$ and
  $abc$ can be listed only in the call to \ListSubtreeConflicts in line~\ref{ln:subtreeQ1} for $(x_p,x_q)=(u_p,u_q)$.
  If $abc<_Q v_q$, then $ab\subseteq\unc{v_q}{v_p}$ and
  $abc$ can be listed only in the call to \ListSubtreeConflicts in line~\ref{ln:subtreeQ2} for $(x_p,x_q)=(v_p,v_q)$.
  The case that $\LCA[P]{abc}\ne r(P)$ and $\LCA[Q]{abc}=r(Q)$ is completely analogous.
  Since the case that $\LCA[P]{abc}\ne r(P)$ and $\LCA[Q]{abc}\ne r(Q)$ is treated in a different recursive step, this case distinction is exhaustive
  and $abc$ is indeed not listed twice.
  
  \smallskip
  To show the running time, let $\T$ denote the recursion tree for input $(P,Q)$ and,
  for each node $v$ of $\T$, let $\delta_v$ and $\gamma_v$ denote the time spent
  in lines~\ref{ln:choose childs}--\ref{ln:end const part} and
  in lines~\ref{ln:comp sets}--\ref{ln:subtreeQ2}, respectively.
  Then, the algorithm finishes in $\sum_{v\in V(\T)}\left(\delta_v+\gamma_v+O(1)\right)$ time.
  First, using the leaf-set equivalence relation computed in line~\ref{ln:leaf equiv} in the parent of $v$
  (or pre-computed if $v$ is the root),
  we execute lines~\ref{ln:choose childs}--\ref{ln:end const part} in constant time, that is, $\delta_v\in O(1)$.
  Second, by the consideration above, tasks~\eqref{it:list}--\eqref{it:prep} can be completed in $O(d_r)$ time,
  where $d_r$ is the number of triples output by \ListCommonRootConflicts and \ListSubtreeConflicts,
  that is, in lines~\ref{ln:common}--\ref{ln:subtreeQ2}.
  Then, we can bound the total running time by
  \[
    \smashoperator[r]{\sum_{v\in V(\T)}}\delta_v + \smashoperator[r]{\sum_{v\in V(\T)}}\gamma_v = O(|\T|) + O(\sum d_r) = O(|\T|+d)
  \]
  where $\sum d_r = d$ because each conflict has a root and no conflict is listed twice (see Lemma~\ref{lem:LSC correct}).
  Finally,
  note that the leaf-sets of the recursive calls of \ListAllConflicts' form a partition of $X$ and,
  therefore, each leaf of $\T$ has a ``private'' element of $X$ that occurs only in that leaf, implying $|\T|\in O(|X|)$.
\end{proof}

\begin{theorem}\label{thm:result}
  Given phylogenetic trees $P$ and $Q$ on the same set of $n$ taxa,
  \autoref{alg:final} enumerates all $d$ conflict triples in $O(n+d)$~time.
\end{theorem}

\section{Conclusion}\label{sec:conc}

We have shown how to list all conflict triples between two phylogenetic trees in $O(n+d)$~time
where $n$ is the number of taxa and $d$ is the number of listed conflicts.
This improves the previously used, trivial $\Theta(n^3)$-time algorithm that tests for each leaf-triple $abc$ for being a conflict.
The presented algorithm is fastest-possible (up to constant factors), since all algorithms solving the problem
must at least read the input and write the output.

A simple next step is to extend the algorithm to non-binary outbranchings.
More challengingly,
we want to reconsider other polynomial-time enumeration problems parameterized by the length of the output list
in hope to produce more ``fastest-possible'' algorithms.
We also plan to analyze real-world phylogenetic trees to see whether the parameter is sufficiently smaller than $n^3$
to make it worth implementing in practice.

\paragraph*{Acknowledgments}
I thank the \emph{Institut de Biologie Computationelle} for funding my research,
as well as my colleagues Krister Swenson and Celine Scornavacca for fruitful discussions.

\let\oldthebibliography=\thebibliography
\let\endoldthebibliography=\endthebibliography
\renewenvironment{thebibliography}[1]{%
  \begin{oldthebibliography}{#1}%
    \setlength{\parskip}{0ex}%
    \setlength{\itemsep}{0ex}%
    \footnotesize
}{
    \end{oldthebibliography}%
}
\bibliographystyle{abbrvnat}
\bibliography{trip.bib}

\end{document}